\documentclass[12pt]{iopart}
\expandafter\let\csname equation*\endcsname\relax
\expandafter\let\csname endequation*\endcsname\relax
\usepackage{
    amsfonts,
    amsmath,
    amssymb,
    amsthm,
    bm,
    braket,
    breakurl,
    doi,
    dsfont,
    graphicx,
    hyperref,
    latexsym,
    mathpazo,
    mathrsfs
}

\newtheorem*{thresh}{Threshold Theorem}
\newtheorem{prop}{Proposition}
\theoremstyle{definition}
\newtheorem{defn}{Definition}
\newtheorem{ex}{Example}

\begin{document}

\title{
  Bounding quantum gate error rate
  based on reported average fidelity
}

%% AUTHORS %%%%%%%%%%%%%%%%%%%%%%%%%%%%%%%%%%%%%%%%%%%%%%%%%%%%%%%%%%%%%%%%%%%%%

\author{%
    Yuval R. Sanders$^{1,2}$,
    Joel J. Wallman$^{1,3}$ and
    Barry C. Sanders$^{1,4,5,6,7}$
}

\address{$^1$%
    Institute for Quantum Computing,
    University of Waterloo,
    Waterloo, Ontario, Canada N2L 3G1
}

\address{$^2$%
    Department of Physics and Astronomy,
    University of Waterloo,
    Waterloo, Ontario, Canada N2L 3G1
}

\address{$^3$%
    Department of Applied Mathematics,
    University of Waterloo,
    Waterloo, Ontario, Canada N2L 3G1
}

\address{$^4$%
    Institute for Quantum Science and Technology,
    University of Calgary,
    Calgary, Alberta, Canada T2N 1N4
}

\address{$^5$%
    Hefei National Laboratory for Physics at Microscale
    and Department of Modern Physics,
    University of Science and Technology of China,
    Hefei, Anhui, China
}

\address{$^6$%
    Shanghai Branch,
    CAS Center for Excellence and Synergetic Innovation Center
    in Quantum Information and Quantum Physics,
    University of Science and Technology of China,
    Shanghai 201315, China
}

\address{$^7$%
    Program in Quantum Information Science,
    Canadian Institute for Advanced Research,
    Toronto, Ontario M5G 1Z8, Canada
}
    
\ead{ysanders@uwaterloo.ca}

%% EOF %%%%%%%%%%%%%%%%%%%%%%%%%%%%%%%%%%%%%%%%%%%%%%%%%%%%%%%%%%%%%%%%%%%%%%%%%
\date{\today}

%% ABSTRACT%%%%%%%%%%%%%%%%%%%%%%%%%%%%%%%%%%%%%%%%%%%%%%%%%%%%%%%%%%%%%%%%%%%%%

\begin{abstract}
Remarkable experimental advances in quantum computing are exemplified by recent
announcements of impressive average gate fidelities exceeding 99.9\% for
single-qubit gates and 99\% for two-qubit gates. Although these high numbers
engender optimism that fault-tolerant quantum computing is within reach, the
connection of average gate fidelity with fault-tolerance requirements is not
direct. Here we use reported average gate fidelity to determine an upper bound
on the quantum-gate error rate, which is the appropriate metric for assessing
progress towards fault-tolerant quantum computation, and we demonstrate that
this bound is asymptotically tight for general noise. Although this bound is
unlikely to be saturated by experimental noise, we demonstrate using explicit
examples that the bound indicates a realistic deviation between the true error
rate and the reported average fidelity. We introduce the Pauli distance as a
measure of this deviation, and we show that knowledge of the Pauli distance
enables tighter estimates of the error rate of quantum gates.
\end{abstract}

%% EOF %%%%%%%%%%%%%%%%%%%%%%%%%%%%%%%%%%%%%%%%%%%%%%%%%%%%%%%%%%%%%%%%%%%%%%%%%

\pacs{03.67.Lx, 03.67.Pp, 03.67.Ac}

\maketitle

%% INTRODUCTION %%%%%%%%%%%%%%%%%%%%%%%%%%%%%%%%%%%%%%%%%%%%%%%%%%%%%%%%%%%%%%%%

\section{Introduction}
\label{sec:intro}

An ideal quantum computer could outperform any classical computer for certain
computational problems in the sense that resource costs such as time or space
scale better than for the best-known classical algorithm~\cite{NC00}. Famous
examples include the provable quadratic speedup for search~\cite{BBHT99}, the
presumed exponential speedup for the Abelian Hidden Subgroup
Problem~\cite{Sim94, Sho99}, and the possible speedup for stoquastic
Hamiltonians using adiabatic quantum computing~\cite{BDOT08}. If a problem
instance with an $\ell$-bit input is solved within bounded error $\varepsilon$
by an algorithm employing $n$ bits or qubits (space cost) and $\nu$ Boolean or
unitary gates (time cost), the algorithm is considered efficient if $n, \nu \in
\text{poly} (\ell)$~\cite{Cob65} and,
if $\varepsilon$ is treated as an asymptotic variable and not as a constant,
$n, \nu \in \text{polylog} (1/\varepsilon)$~\cite{DN06}.

In practice, preparation, processing and measurement are faulty, but the
Threshold Theorem for Fault-Tolerant Quantum Computation (``Threshold
Theorem'')~\cite{Sho95, AB97, KLZ98, AB08, DMN13} guarantees that a noisy device
can perform scalable fault-tolerant quantum computations under certain
conditions. Specifically, the Threshold Theorem guarantees the existence of a
threshold error rate $\eta_0$ ($0<\eta_0<1$) such that a faulty computer whose
error rate $\eta$ satisfies $\eta < \eta_0 $ can perform universal quantum
computations efficiently, namely with $\text{polylog} (n, \nu)$ additional
overhead. The Threshold Theorem is the key to establishing that faulty quantum
computers can be as efficient as ideal quantum computers. A key drawback of the
Threshold Theorem is that $\eta_0$ is established existentially, not
constructively~\cite{AB08}; consequently, this scalability figure of merit is
elusive in practice. A practical approach to assessing fault-tolerance is to
establish a lower bound $\eta_0^{\rm lb} \leq \eta_0$ by devising a code that
is robust against errors that occur at a rate lower than $\eta_0^{\rm lb}$; the
C4/C6 code, for example, is
known to have a threshold of $\eta_0^{\rm lb} \leq 3\%$~\cite{Kni05}.

Current experimental characterizations of quantum gates do not report $\eta$.
Instead the average gate fidelity~$\varphi$~\cite{Nie02} is the typical figure
of merit for gate performance because it can be reliably and scalably estimated
using a procedure called randomized benchmarking~\cite{EAZ05}. Recent reports of
$\varphi$ exceeding $99.9\%$ for one-qubit gates and $99\%$ for two-qubit
gates~\cite{BKM+14} generate strong optimism about the feasibility of scalable
quantum computing. But despite its experimental convenience, $\varphi$ is not
the correct quantity to assess scalability via the Threshold Theorem.

Our aim is to convert reported~$\varphi$ to an upper bound~$\eta^{\rm ub}$
for~$\eta$. This upper bound provides a sufficient condition for fault-tolerant
quantum computing: errors can be efficiently corrected if
\begin{equation}
  \label{eq:condition}
  \eta \leq \eta^{\rm ub} < \eta_0^{\rm lb} \leq \eta_0.
\end{equation}
Thus, a code-derived $\eta_0^{\rm lb}$ can be used to determine whether the
fidelity-derived $\eta^{\rm ub}$ suffices for scalability. The quantity
$\eta^{\rm ub}$ can therefore be used to assess scalability based upon the
experimentally convenient average fidelity~$\varphi$.

Given the fidelity $\varphi$, the best known upper bound $\eta^{\rm ub}$ is
\begin{equation}
  \eta^{\rm ub} := d\sqrt{(1+d^{-1})(1-\varphi)},
\end{equation}
where $d$ is the dimension of the system being acted on~\cite{BK11,Mag12,WF14}.
This bound is unfortunate because the square-root ensures that superficially 
impressive gate fidelities do not, by themselves, guarantee high-quality gate 
performance. A two-qubit gate with $99\%$ fidelity is, for example, only 
guaranteed to have an error rate below $45\%$. Indeed, we have an explicit 
example (Example~\ref{ex:high_fid_high_error_2}) of a two-qubit gate with 
fidelity $99\%$ but an error rate slightly under $13\%$. Furthermore, we 
demonstrate (Example~\ref{ex:dep_vs_unit}) that assessments of gate performance 
based on fidelity can mislead about the relative importance of different noise 
sources.

Our main claim is that $\eta^{\rm ub}$ is an asymptotically-tight approximation 
to the least upper bound to $\eta$ in terms of $\varphi$ and $d$. The least
upper bound is a function $\eta^{\rm lub} (\varphi, d)$ satisfying the following
two properties:
\begin{enumerate}
  \item for any noise channel acting on a $d$-dimensional system with average 
  fidelity $\varphi$ and error rate $\eta$, $ \eta \leq
  \eta^{\rm lub} (\varphi, d) $; and
  \item $ \eta^{\rm lub} (\varphi, d) \leq f(\varphi, d)$ for any function 
  $f(\varphi, d)$ satisfying the first property.
\end{enumerate}
We show that $\eta^{\rm lub}$ must scale as $\sqrt{1 - \varphi}$ for fixed $d$ 
(Proposition~\ref{prop:fidelity_scaling}) and must scale as $d$ for fixed 
$\varphi$ (Proposition~\ref{prop:dim_scaling}). We conjecture that
$\eta^{\rm ub} = \eta^{\rm lub}$.

We suggest one potentially useful kind of additional information about gate 
performance: a quantity we call the ``Pauli distance'' $\delta^{\rm Pauli}$. 
This quantity is motivated by the fact that Pauli channels with average 
fidelity $\phi$ have an error rate of $\eta^{\rm Pauli} = (1 + d^{-1})
(1-\varphi)$, saturating a lower bound on the error rate in terms of the
fidelity and dimension~\cite{MGE12,WF14}. We show that an arbitrary noise
channel satisfies $\eta \leq \eta^{\rm Pauli} + \delta^{\rm Pauli}$
(Proposition~\ref{prop:Pauli_distance}), so that smaller upper bounds
$\eta^{\rm ub}$ to the error rate $\eta$ of a quantum gate that avoid the 
$\sqrt{1-\varphi}$ scaling can be found if the noise is promised to be nearly 
Pauli.

Our message is not that impressive reported average gate fidelities
fail to demonstrate real progress towards fault-tolerant quantum computing,
but that these reports are \emph{insufficient} to claim that fault tolerance
is now within reach. Our argument is that reported fidelity alone implies only
loose bounds on the quantum gate error rate, and that
tighter bounds on error rate are possible only if performance metrics
other than average fidelity are also considered. In addition to our suggestion
of the Pauli distance as a useful additional figure-of-merit,
an intriguing quantity known as ``unitarity''~\cite{WGHF15,KLDF15,Wal15}
may also prove to be useful for assessing the performance of quantum gates.

Our paper proceeds as follows. We establish the definition of error rate in
Sec.~\ref{sec:errors}. We give a brief review of the average gate fidelity and
its relationship to the error rate in Sec.~\ref{sec:fidelity}. Our asymptotic
tightness result is presented in Sec.~\ref{sec:tightness}, and we introduce the
Pauli distance in Sec.~\ref{sec:Pauli_distance}. We use our results in
Sec.~\ref{sec:discussion} to assess reported progress towards fault-tolerance,
and we conclude in Sec.~\ref{sec:conclusion}.

%% ERRORS %%%%%%%%%%%%%%%%%%%%%%%%%%%%%%%%%%%%%%%%%%%%%%%%%%%%%%%%%%%%%%%%%%%%%%

\section{Error rates and the Threshold Theorem}
\label{sec:errors}

The Threshold Theorem is currently our only rigorous guarantee that
fault-tolerant quantum computing is viable if threshold operating conditions are
met. The threshold operating conditions take two forms: noise must be restricted
to a promised form and measurable errors must occur at a low enough rate. We
first define the error rate of quantum gates by extension from
the error rate of random processes (Sec.~\ref{subsec:gate_error}). We then
explain the Threshold Theorem and its connection
to our definition of error rate and to numerical estimates
of code-specific threshold bounds (Sec.~\ref{subsec:threshold_thm}).

% ------------------------------------------------------------------------------

\subsection{The error rate of a quantum logic gate}
\label{subsec:gate_error}

Our definition of the error rate of a quantum logic gate builds naturally on 
the concept of error rate for a random process, that is, a map from input to 
output states. For deterministic processes, we can say that an error has 
occurred if the process produces the `wrong' output. However, no single output 
of a random process can be treated unambiguously as `correct'. We therefore 
define the rate of error for a process by comparing the actual statistics of 
the process to its ideal statistics. 

The statistics for an ideal process is governed by a probability distribution 
$p_{\rm id}$ over the set of possible outputs $X$; the ideal probability of 
output $x\in X$ is $p_{\rm id}(x)$. An error-prone process produces a different 
distribution $p_{\rm ac}$, governing the actual statistics over the set of 
possible outcomes $X$. The total variation distance 
\begin{equation}
  d_{\rm TV} \left( \mu, \nu \right) \equiv
  \frac{1}{2}\sum_{x\in X} \left\vert \mu(x) - \nu(x) \right\vert
\end{equation} 
is a natural measure of the distance between two probability distributions 
$\mu$ and $\nu$ over a set of outcomes $X$. 

The total variation distance $d_{\rm TV} \left( p_{\rm ac}, p_{\rm id} \right)$
can be interpreted as the error rate of the process as follows. We can estimate
$p_{\rm ac}$ by sampling the actual random process $N$ times and counting the
number $n(x)$ of occurrences of each possible output $x$; the fraction $n(x)/N$
approaches $p_{\rm ac}(x)$ as $N \rightarrow \infty$. By altering some fraction
$r$ of the samples so that the number of occurrences of each outcome $x$ becomes
$n^\prime (x)$, we can ensure that $n^\prime(x)/N \approx p_{\rm id}(x)$ rather
than $p_{\rm ac}(x)$. The fraction $r$ is not unique, but the minimum possible
value $r_{\rm min}$ of $r$ must be greater than zero for large $N$ if
$p_{\rm ac} \neq p_{\rm id}$. By Proposition~4.7 of~\cite{LPW09}, $r_{\rm min}
\rightarrow d_{\rm TV} \left( p_{\rm ac}, p_{\rm id} \right)$ as $N \rightarrow
\infty$. Thus, $d_{\rm TV} \left( p_{\rm ac}, p_{\rm id} \right)$ approximates
the fraction of a large sample that must be altered to ensure that the relative
frequencies of each outcome match the ideal distribution $p_{\rm id}$; each
alteration can be interpreted as the correction of an error.

We claim that the total variation distance induces the diamond distance
$d_\diamond$ on the space of quantum channels. Our argument follows from the
work of Fuchs and van de Graaf~\cite{FvdG99}, which shows that the error rate of
quantum states is given by the trace distance between the quantum states, and
the work of Kitaev~\cite{Kit97}, which shows that the diamond norm extends the
trace norm to quantum channels. We begin by defining some terminology.

Quantum logic gates act reversibly on some fixed quantum register.
Ideally, the state of
this register can be represented by a unit vector in a fixed $d$-dimensional
Hilbert space $\mathscr{H}$. The register is typically treated as a collection
of $n$ qubits, in which case $\mathscr{H}$ is canonically isomorphic with the
$n$-fold tensor product of the Hilbert space $\mathscr{Q}\cong\mathbb{C}^2$ of a
single qubit: $\mathscr{H} \cong \mathscr{Q}^{\otimes n}$. In this case, $d =
2^n$.

Whereas ideal register states are represented by unit vectors $\ket{\psi} \in
\mathscr{H}$, realistic states are modelled by density operators $\rho$ on
$\mathscr{H}$. A measurement of a quantum state is described by a positive
operator-valued measure (POVM), which is a set of positive operators
$\left\{ E_\ell \right\}$ acting on $\mathscr{H}$ such that $\sum_\ell E_\ell =
I$, the identity operator. The probability of observing the outcome labelled
$\ell$ is $\Tr \left(E_\ell \rho\right)$. Thus, the actual output
$\rho_{\rm ac}$ of a gate acting on a specified input can be compared to the
ideal output $\rho_{\rm id}$ by measuring with respect to some POVM. The error
rate of this measurement is $d_{\rm TV} \left(p_{\rm ac},p_{\rm id}\right)$,
where $p_{\rm ac} (\ell) = \Tr \left( E_\ell \rho_{\rm ac} \right)$ and
$p_{\rm id} (\ell) = \Tr \left( E_\ell \rho_{\rm id} \right)$. Maximizing
$d_{\rm TV} \left( p_{\rm ac}, p_{\rm id} \right)$ over all possible choices of
measurements yields~\cite{FvdG99}
\begin{equation}
  \label{eq:operator_distance}
  d_{\rm Tr} \left( \rho_{\rm ac}, \rho_{\rm id} \right) :=
  \frac{1}{2} \left\| \rho_{\rm ac} - \rho_{\rm id} \right\|_{\rm Tr},
\end{equation}
where $\left\| A \right\|_{\rm Tr} := \Tr \sqrt{ A^\dagger A }$ for any linear
operator $A$. Thus, the error rate of $\rho_{\rm ac}$ with respect to
$\rho_{\rm id}$ is $d_{\rm Tr} \left( \rho_{\rm ac}, \rho_{\rm id} \right)$.

Now that we have defined the error rate of the output of a quantum logic gate,
we can define the error rate of the gate itself. We first present some
mathematical notation for evaluating the difference between an ideal and real
implementation of a quantum logic gate.

An ideal quantum logic gate, represented by $G$, acts as a unitary operator on
$\mathscr{H}$. Whereas the operation on a pure state can be treated as direct
(i.e.~$\ket{\psi} \mapsto G\ket{\psi}$), the gate can act upon a mixed state
$\rho$. In this instance, the gate performs the action $\rho \mapsto G\rho
G^\dagger$. This action is represented by a quantum channel
$\mathcal{G}_{\rm id}$; explicitly,
\begin{equation}
  \label{eq:ideal_gate}
  \mathcal{G}_{\rm id} (\rho) := G \rho G^\dagger.
\end{equation}
This channel is compared with a non-ideal implementation $\mathcal{G}_{\rm ac}$
that is in general not represented by unitary conjugation but is a completely
positive, trace preserving linear operator on the space of density operators
over $\mathscr{H}$.

We have established in Eq.~(\ref{eq:operator_distance}) that the error rate for
a quantum logic gate acting on a specified input state $\rho$ is given by
$d_{\rm Tr} \left(
  \mathcal{G}_{\rm ac}(\rho) , \mathcal{G}_{\rm id}(\rho)
\right)$. The error rate of $\mathcal{G}_{\rm ac}$ with respect to
$\mathcal{G}_{\rm id}$ involves maximization over inputs. Whereas the error rate
could be defined as $\max_\rho d_{\rm Tr} \left(
  \mathcal{G}_{\rm ac}(\rho) , \mathcal{G}_{\rm id}(\rho)
\right)$, such a definition is undesirable because, in general, the error rate
of $\mathcal{G}_{\rm ac}\otimes\mathds{1}$ (where $\mathds{1}$ acts on some
ancillary space $\mathscr{H}^\prime$) differs from that of
$\mathcal{G}_{\rm ac}$~\cite{Kit97,Wat11}. We therefore amend this definition by
maximizing over inputs and ancillary spaces using a construction called the
diamond norm~\cite{Kit97}:
\begin{equation}
  \left\| \mathcal{A} \right\|_\diamond := 
  \sup_{\mathscr{H}^\prime}
  \sup_{ \rho \in \mathrm{dens} ( \mathscr{H} \otimes \mathscr{H}^\prime ) }
  \left\| \mathcal{A} \otimes \mathds{1} (\rho) \right\|_{\rm Tr},
\end{equation}
where $\mathcal{A}$ is any superoperator over $\mathscr{H}$ and $\mathrm{dens}
(\mathscr{H} \otimes \mathscr{H}^\prime)$ is the set of density operators over
the joint Hilbert space of the original register and some ancilla. We therefore
define
\begin{equation}
  d_\diamond \left( \mathcal{G}_{\rm ac}, \mathcal{G}_{\rm id} \right) :=
  \frac{1}{2} \left\|
    \mathcal{G}_{\rm ac} - \mathcal{G}_{\rm id}
  \right\|_\diamond
\end{equation}
to be the error rate $\eta$ of $\mathcal{G}_{\rm ac}$ with respect to
$\mathcal{G}_{\rm id}$: $\eta = d_\diamond \left(
  \mathcal{G}_{\rm ac}, \mathcal{G}_{\rm id}
\right)$. However, we shall use a
modified but equivalent form of this definition in the remainder of this paper.
Our modification is purely for mathematical convenience. 

\begin{defn}
If $\mathcal{G}_{\rm ac}$ is some implementation of a gate $G$, define
\begin{equation}
  \mathcal{D}_G := \mathcal{G}_{\rm ac}\circ\mathcal{G}_{\rm id}^{-1}
\end{equation}
to be the \emph{discrepancy channel} of $G$, where the channel
$\mathcal{G}_{\rm id}$ defined in Eq.~(\ref{eq:ideal_gate}) is unitary and hence
invertible.
\end{defn}

\begin{defn}
The \emph{error rate} of an implementation of $G$ is given by
\begin{equation}
  \label{eq:error_rate_via_discrepancy}
  \eta = d_\diamond \left( \mathcal{D}_G, \mathds{1} \right),
\end{equation}
where $\mathcal{D}_G$ is the discrepancy channel of the implementation.
\end{defn}

% ------------------------------------------------------------------------------

\subsection{The Threshold Theorem}
\label{subsec:threshold_thm}

We now explain the Threshold Theorem. We begin by elaborating on the promised 
form of noise; namely, noise locality. We then identify a statement of the 
theorem that is appropriate for our needs. Finally, we
review commonly quoted estimates of fault-tolerance thresholds.

We elaborate on the assumption of noise locality because it is required for
defining the error rate of logic gates independent of the circuit in which they
are employed. Briefly, a logic circuit is said to experience local noise if the
noise acts separately on individual logic gates. To be precise, recall that a
logic circuit is defined as a directed acyclic graph with nodes labelled by
elements of some set of logic gates, where arrows into a node represent inputs
and arrows out of a node represent outputs~\cite{Sav97}. A quantum logic circuit
can similarly be represented by a directed acyclic graph. The noise of a logic
circuit is local if it can be represented as the composition of noise processes
on individual nodes of the circuit graph.

As noise is assumed to affect each gate independently, we model noise by
replacing the intended unitary gate $G$ by some imperfect implementation
$\mathcal{G}_{\rm ac}$ represented as a quantum channel (i.e.~a completely
positive, trace-preserving linear map on density operators over the state space
of the input register). Such a model is reasonable for an imperfect gate subject
to local noise if the interaction between the register space and its environment
obeys the Born-Markov approximation~\cite{Lin76}. We therefore assign an error
rate $\eta$ to each gate $\mathcal{G}_{\rm ac}$ in a circuit $Q^\prime$, which
simulates $Q$ efficiently and accurately in the presence of local noise if $\eta
< \eta_0$.

The various formulations of the Threshold Theorem are distinguished by
assumptions concerning noise. We prefer to employ the statement of Aharonov and
Ben-Or because of its directness with a minimum of jargon.

\begin{thresh}[\cite{AB08}]
There exists a threshold $\eta_0 > 0$ such that the following holds.
Let $\varepsilon > 0$.
If $Q$ is a quantum circuit
operating on $n$ input qubits for $t$ time steps
using $s$ two- and one-qubit gates,
there exists a quantum circuit $Q^\prime$
with depth, size, and width overheads
which are polylogarithmic in $n$, $s$, $t$, and $1/\varepsilon$
such that, in the presence of local noise of error rate $\eta < \eta_0$,
$Q^\prime$ computes a function
which is within $\varepsilon$ total variation distance
from the function computed by $Q$.
\end{thresh}

This theorem guarantees that a value $\eta_0 > 0$, called the ``threshold'',
exists such that a quantum circuit $Q$ can be efficiently simulated by another
circuit $Q^\prime$ to within an arbitrary error tolerance $\varepsilon > 0$ even
if $Q^\prime$ is subject to ``local noise'' at a rate $\eta < \eta_0$.
Inequivalent statements of `the' Threshold Theorem are inequivalent because they
assume promises about noise that are different from that of noise locality.

There are two important limitations to the utility of the threshold $\eta_0$.
Firstly, surpassing the threshold is sufficient but not necessary for
fault-tolerance: error rates larger than $\eta_0$ could be acceptable if
stronger promises can be made about noise. Conversely, devices subject to noise
that does not satisfy the assumptions of the Threshold Theorem cannot be said to
be fault-tolerant based on a demonstration that error rates fall below
threshold; a stronger threshold $\eta_0^\prime < \eta_0$ could apply. The second
limitation is that the choice of $Q^\prime$ depends in practice upon the
specified quantum-error-correcting code. Based on the choice of code, the
appropriate performance target is $\eta_0^{\rm lb}$, rather than $\eta_0$. As
with the first limitation, the validity of $\eta_0^{\rm lb}$ as a performance
target derives from the validity of the promises made about the noise affecting
real devices.

Whereas some threshold
estimates are obtained through rigorous analysis of the performance of a code in
the presence of noise subject to promises of varying strength, others are
obtained through numerical simulation of performance in the presence of a
parametrized family of noise models. Estimates based on numerical simulation are
more optimistic and are often used as performance targets for experimental
fault-tolerant quantum computing research~\cite{BKM+14,CGC+12}.

Analytic estimates of the threshold can be produced based on details of the
proof of the Threshold Theorem. Aharonov and Ben-Or, for example, can justify an
estimate of $\eta_0^{\rm lb} \approx 10^{-6}$~\cite{AB08} based on their choice
of coding strategy. They report a value of $\eta_0^{\rm lb} \approx
10^{-3}$~\cite{Rei06,AGP08} as being the largest rigorously established value.
Numerical estimates, by contrast, are produced by simulating the behaviour of an
error-correcting code in the presence of a restricted class of noise models,
typically depolarizing~\cite{Kni05,RH07}. The relationship of these estimates
with thresholds of the kind established by the Threshold Theorem is not
clear~\cite{STD05,SCCA06}, but these simulations are nonetheless often seen as
indicative~\cite{FMMC12} of true threshold values. The surface code, in
particular, is often believed to have a threshold of around
$1\%$~\cite{RH07,FMMC12}.

A direct comparison of the above threshold values is not justifiable because
each value makes different assumptions about the behaviour of noise and the
choice of code. Thus, the estimate of $\eta_0^{\rm lb} \approx 1\%$ for the
surface code under depolarizing noise does not make the surface code less
desirable than Knill's C4/C6 code even though
the C4/C6 code could have a threshold as high as $3\%$~\cite{Kni05}
because there are other
practical reasons to prefer the surface code over the C4/C6 code. Similarly,
actual gate performance should not be directly compared with these threshold
values because those gates are certainly subject to noise that is not
well-approximated by the depolarizing noise model. The connection between
numerical simulations and fault-tolerance thresholds is a matter of active
research~\cite{MPGC13, PGH+14}.

Whereas there are important open questions regarding the interpretation of
threshold estimates produced by simulation, the term `threshold' is
unambiguously a reference to an upper bound of the error rate introduced by any
given logic gate in a quantum circuit. The main point of this paper is to
connect these theoretical characterizations of error to the experimentally
convenient average gate fidelity.

%% FIDELITY %%%%%%%%%%%%%%%%%%%%%%%%%%%%%%%%%%%%%%%%%%%%%%%%%%%%%%%%%%%%%%%%%%%%

\section{Average gate fidelity}
\label{sec:fidelity}

Whereas $\eta_0$, defined by the Threshold Theorem, and $\eta_0^{\rm lb}$,
defined in Eq.~(\ref{eq:condition}) and established by noise models and coding
strategies, are appropriate quantities for analyzing scalability, average gate
fidelity is employed in experimental studies because of its convenience. In this
section, we define average gate fidelity and discuss the connection between
average gate fidelity and error rate, first by reviewing the literature and then
by constructing an example that shows that this connection is problematic:
average gate fidelity and quantum gate error rate are not directly connected.
Instead, only lower and upper bounds to the error rate can be derived from
fidelity; the gap between these bounds is substantial in regimes of interest.

The fidelity of a state $\rho$ to a pure state $\psi$ is $\Tr \left( \ket{\psi}
\! \bra{\psi} \rho \right)$~\cite{EAZ05,MGE11,ECMG14,KdSR+14}. The fidelity  of
the output of the actual gate $\mathcal{G}_{\rm ac}$ to the output of the ideal
gate $\mathcal{G}_{\rm id}$ for a given input state $\ket{\psi}$ is therefore
\begin{align}
  \Tr \left(
    \mathcal{G}_{\rm id} (\ket{\psi} \! \bra{\psi})
    \mathcal{G}_{\rm ac}(\ket{\psi} \! \bra{\psi})
  \right) =
  \bra{\psi} \mathcal{D}_G \left( \ket{\psi} \! \bra{\psi} \right) \ket{\psi}.
\end{align}
Averaging over pure state inputs with respect to the Haar measure then gives the
average gate fidelity
\begin{equation}
  \label{eq:fidelity}
  \varphi :=
  \int \text{d}\mu(\psi) \bra{\psi}
      \mathcal{D}_G\left( \ket{\psi} \! \bra{\psi} \right)
  \ket{\psi}
\end{equation}
where we have used the unitary invariance of the Haar measure and $\mathcal{D}_G
= \mathcal{G}_{\rm ac} \circ \mathcal{G}_{\rm id}^{-1}$. The popular randomized
benchmarking protocol~\cite{EAZ05} produces an estimate of this quantity
averaged over a gate-set, though proposed extensions~\cite{MGJ+12} produce
estimates of the average gate fidelity for individual gates.

The state fidelity can be interpreted as the error rate of a particular
measurement. If we define for each pure state $\ket{\psi}$ the POVM $\left\{
\ket{\psi}\!\bra{\psi}, \mathds{1} - \ket{\psi}\!\bra{\psi} \right\}$, the
outcome of this measurement applied to a state $\rho$ will be
$\ket{\psi}\!\bra{\psi}$ with probability $\bra{\psi}\rho\ket{\psi}$, which is
the state fidelity of $\rho$ with respect to $\ket{\psi}$. Thus, the total
variation distance of the actual statistics $p_{\rm ac}$ from the ideal
statistics $p_{\rm id}$ of this measurement upon the output
$\mathcal{D}_G (\ket{\psi}\!\bra{\psi})$
is the state infidelity of $\mathcal{D}_{G} (\ket{\psi}\!\bra{\psi})$
with respect to $\ket{\psi}\!\bra{\psi}$. 

However, the average gate infidelity $1-\varphi$ cannot be so easily interpreted
as an average error rate (despite the common habit~\cite{BWC+11,CGC+12,BKM+14}),
as the measurement basis is not fixed in the integral (and so the infidelity is
not an average error for a fixed measurement), yet neither is it averaged
independently from the state.

To clarify the relationship between average gate fidelity and error rate, we
consider two noise processes on a single qubit. The first is given by
depolarizing noise
\begin{equation}
  \label{eq:depolarizing_error}
  \mathcal{E}^{\rm dep}_r(\rho) := (1-r)\rho + r I/2,
\end{equation}
where $I$ is the identity operator on $\mathscr{Q}$, and the second is a unitary
error
\begin{equation}
  \label{eq:unitary_error}
  \mathcal{E}^{\rm U}_\theta(\rho) := U\rho U^\dagger,
\end{equation}
where $U$ is some unitary operator on $\mathscr{Q}$ with eigenvalues $e^{\pm i
\theta}$ for $0 \leq \theta \leq \pi$. The average gate fidelity for
depolarizing noise is
\begin{equation}
  \varphi^{\rm dep}(r) = 1 - \frac{r}{2}
\end{equation}
whereas the fidelity of the unitary error is~\cite{Nie02}
\begin{equation}
  \varphi^{\rm U} (\theta) = \frac{1}{3} + \frac{2}{3} \cos^2\theta.
\end{equation}
By contrast,
\begin{equation}
  \eta^{\rm dep}(r) = \frac{3}{4}r
\end{equation}
for depolarizing noise, which follows from the fact that depolarizing noise is
Pauli~\cite{MGE12}, and
\begin{equation}
  \eta^{\rm U}(\theta) = \sin\theta
\end{equation}
for unitary error~\cite{JKP09}. Therefore,
\begin{equation}
  \eta^{\rm dep} = \frac{3}{2}\left(1-\varphi^{\rm dep}\right)
\end{equation}
for depolarizing noise but
\begin{equation}
  \eta^{\rm U} = \sqrt{ \frac{3}{2} \left( 1 - \varphi^{\rm U} \right) }
\end{equation}
for unitary error,
which means that there is no single function $f$
such that $f(\varphi) = \eta$ for every possible noise channel.
We demonstrate this difficulty in the following example.

\begin{ex}
\label{ex:dep_vs_unit}
Consider a single-qubit gate that is subject to the two noise processes of
Eqs.~(\ref{eq:depolarizing_error}) and~(\ref{eq:unitary_error}): depolarizing
and unitary. The depolarizing rate is $r=10^{-3}$, with corresponding fidelity
\begin{equation}
  \varphi^{\rm dep} = 1 - 5.0\times 10^{-4},
\end{equation}
whereas the unitary error has angle $\theta=10^{-2}$, with corresponding
fidelity
\begin{equation}
  \varphi^{\rm U} = 1 - 6.7 \times 10^{-5}.
\end{equation}
The combination
\begin{equation}
  \mathcal{D}_G :=
  \mathcal{E}^{\rm dep}_r \circ \mathcal{E}^{\rm U}_\theta \equiv
  \mathcal{E}^{\rm U}_\theta \circ \mathcal{E}^{\rm dep}_r \equiv
  (1-r) \mathcal{E}^{\rm U}_\theta + r \mathcal{E}^{\rm dep}_{r=1}
\end{equation}
has fidelity
\begin{equation}
  \varphi^{\rm tot} = (1-r)\varphi^{\rm U} + \frac{r}{2} =
  1 - 5.3 \times 10^{-4},
\end{equation}
so the fidelity loss seems to arise mostly from depolarizing noise. Yet the
error rate due to unitary error is
\begin{equation}
  \eta^{\rm U}=10^{-2}
\end{equation}
whereas the error rate due to depolarizing noise is
\begin{equation}
  \eta^{\rm dep} = 7.5 \times 10^{-4}.
\end{equation}
The triangle inequalities imply that the error rate of the combined noise
process is
\begin{equation}
  \eta^{\rm tot} = (1 \pm 0.08) \times 10^{-2},
\end{equation}
so  the unitary error is in fact dominating over depolarizing even though the
fidelity appears to imply the reverse.
\end{ex}

Thus, information beyond fidelity is needed to assess the relative importance of
various noise processes affecting the quantum computing device. Determination of
the Pauli-distance (Sec.~\ref{sec:Pauli_distance}) is one possible approach to
characterizing the influence of different noise sources; extending randomized
benchmarking to estimate unitarity~\cite{WGHF15,KLDF15} is another.

Although no direct connection between average gate fidelity and error rate
exists in general, average gate fidelity is clearly of some worth: if the
fidelity of a quantum logic gate is precisely one, we are certain that the gate
will always perform exactly as intended. More generally,
\begin{equation}
  \label{eq:Magesan_inequalities}
  \eta^{\rm Pauli} := \left( 1 + d^{-1} \right) (1 - \varphi)
  \leq \eta \leq
  d \sqrt{ \left( 1 + d^{-1} \right) ( 1 - \varphi ) }
  =: \eta^{\rm ub},
\end{equation}
where $d$ is the dimension of $\mathscr{H}$. Note that the upper bound can
exceed unity if
\begin{equation}
  \varphi < 1-(d^2+d)^{-1},
\end{equation}
so a fidelity less than $83\%$ for single-qubit gates or less than $95\%$ for
two-qubit gates does not ensure that $\eta < 1$; it is possible that the gate is
performing incorrectly all the time for at least one input. Ensuring that $\eta
< 1$ when $\varphi$ fails to meet this threshold for non-triviality must involve
additional promises about the form of noise.

To illustrate the gap between the above lower and upper bounds, consider a
target fidelity of $99\%$. Then for one, two and three qubits, the above upper
bound gives $25\%, 45\%$ and $85\%$ respectively, whereas the lower bound is
essentially $1\%$. For target fidelities of $99.9\%$, the upper bounds become
$7.75\%$, $14.2\%$ and $26.9\%$ respectively, whereas the lower bound is
approximately $0.1\%$ in each case. Hence, these upper and lower bounds differ
by orders of magnitude in regimes of experimental interest.

%% TIGHTNESS %%%%%%%%%%%%%%%%%%%%%%%%%%%%%%%%%%%%%%%%%%%%%%%%%%%%%%%%%%%%%%%%%%%

\section{Tightness of the upper bound on the error rate}
\label{sec:tightness}

We now prove that the upper bound $\eta^{\rm ub}$ on the error rate is
asymptotically tight with respect to fidelity for fixed dimension and
asymptotically tight with respect to dimension for fixed fidelity. In addition,
we prove by example that this bound cannot be improved by better than a factor
varying as the square-root of dimension. To demonstrate these facts, we first
define the variables and functions about which we make asymptotic statements.

\begin{defn}
\label{defn:lub}
The \emph{least upper bound of error rate with respect to average gate fidelity}
$\eta^{\rm lub} = \eta^{\rm lub} (\varphi, d)$ is the unique function of
$\varphi$ and $d$ that satisfies the following. For any discrepancy channel
$\mathcal{D}_G$ of dimension $d$ with average gate fidelity $\varphi$ and error
rate $\eta$, $\eta^{\rm lub}(\varphi,d)\geq\eta$. Furthermore, suppose that
$\eta^{\rm ub} = \eta^{\rm ub} (\varphi, d)$ is any other function with the same
property. Then $\eta^{\rm lub} (\varphi, d) \leq \eta^{\rm ub} (\varphi, d)$ for
all $\varphi$ and $d$.
\end{defn}

We shall establish the scaling of $\eta^{\rm lub}(\varphi,d)$ as a function of
each variable when the other is fixed. Notationally, we distinguish fixed from
variable quantities as follows. If the dimension $d$ is fixed but $\varphi$ is
variable, we write $\eta^{\rm lub} (\varphi) \vert_d$; if vice versa,
$\eta^{\rm lub} (d) \vert_\varphi$. We use similar notation for $\eta^{\rm ub}$.
We seek to establish the scaling of $\eta^{\rm lub}$ in the limit
$\varphi\rightarrow 1$. To make asymptotic arguments about this scaling, we
define two variables that go to infinity as $\varphi\rightarrow 1$.

\begin{defn}
\label{defn:inverse_error_rate}
Define the \emph{inverse error rate} of a quantum logic gate to be $\zeta :=
\eta^{-1}$, where $\eta$ is the error rate of the gate. Thus, $\zeta \rightarrow
\infty$ as $\eta \rightarrow 0$. We also write
\begin{equation}
  \zeta^{\rm lub} := \left(\eta^{\rm lub}\right)^{-1},
  \zeta^{\rm ub} := \left(\eta^{\rm ub}\right)^{-1},
\end{equation}
which are lower bounds to $\zeta$.
\end{defn}

\begin{defn}
\label{defn:inverse_avg_fid}
Define the \emph{inverse average infidelity} of a quantum logic gate to be
\begin{equation}
  \upsilon := (1-\varphi)^{-1}.
\end{equation}
Thus, $\upsilon \rightarrow \infty$ as $\varphi \rightarrow 1$.
\end{defn}

Thus, we can write $\zeta^{\rm lub} = \zeta^{\rm lub} (\upsilon, d)$ and compare
this function to
\begin{equation}
  \zeta^{\rm ub}(\upsilon, d) =  \frac{\sqrt{\upsilon}}{d\sqrt{1+\frac{1}{d}}}.
\end{equation}
We show that
\begin{align}
  \zeta^{\rm lub} (\upsilon) \vert_d \in \Theta (\sqrt{\upsilon})
\end{align}
and
\begin{equation}
  \zeta^{\rm lub}(d)\vert_\upsilon \in \Theta (d^{-1});
\end{equation}
thus, $\zeta^{\rm ub}$ has optimal scaling with respect to $\phi$ and $d$ when
the the other is fixed. We shall make use of a particular unitary gate, defined
below.

\begin{defn}
\label{defn:c-phase}
Define the \emph{generalized controlled-phase gate} by
\begin{equation}
  \mathcal{G}_{\rm id} (\rho) := U_\theta \rho U_\theta^\dagger;\ 
  U_\theta := \left(
    \begin{matrix}
       1    &    0    & \cdots  &    0    &    0    \\
       0    &    1    & \cdots  &    0    &    0    \\
    \vdots  & \vdots  & \ddots  & \vdots  & \vdots  \\
       0    &    0    & \cdots  &    1    &   0     \\
       0    &    0    & \cdots  &    0    & e^{i\theta}
    \end{matrix}
  \right),
\end{equation}
where $U_\theta$ is expressed in the computational basis.
\end{defn}

\begin{prop}
\label{prop:fidelity_scaling}
For fixed dimension $d$,
\begin{equation}
  \zeta^{\rm lub} (\upsilon) \vert_d \in \Theta(\sqrt{\upsilon}).
\end{equation}
Therefore,
\begin{equation}
  \zeta^{\rm ub} (\upsilon) \vert_d \in
  \Theta \left( \zeta^{\rm lub}(\upsilon) \vert_d \right).
\end{equation}
Furthermore,
\begin{equation}
  \eta^{\rm lub} (\varphi) \vert_d \geq
  \frac{1}{2} (d-1)^{-\frac{1}{2}} \times \eta^{\rm ub} (\varphi) \vert_{d}.
\end{equation}
\end{prop}

\begin{proof}
Suppose we have an implementation of the generalized controlled-phase gate given
simply by $\mathcal{G}_{\rm ac}(\rho) = \rho$, the identity channel. The average
gate fidelity is~\cite{Nie02}
\begin{equation}
  \label{eq:fidelity_scaling_fid}
  \varphi
   = \frac{d + \left\vert \Tr \left( U_{-\theta} \right) \right\vert^2}{d + d^2}
   = 1 - \frac{2 (d-1)}{d(d+1)} (1 - \cos\theta).
\end{equation}
By Theorem~26 of~\cite{JKP09},
\begin{equation}
  \label{eq:fidelity_scaling_error_rate}
  \eta = \sqrt{ \frac{1 - \cos\theta}{2} };
\end{equation}
hence,
\begin{equation}
  \zeta = \sqrt{\frac{4 (d-1)}{d (d+1)}} \times \sqrt{\upsilon}.
\end{equation}
By contrast,
\begin{equation}
  \zeta^{\rm ub} = \sqrt{\frac{1}{d(d+1)}} \times \sqrt{\upsilon}.
\end{equation}
Furthermore, $\zeta^{\rm lub}$ is defined so that
$\zeta \geq \zeta^{\rm lub} \geq \zeta^{\rm ub}$;
thus,
\begin{equation}
  \label{eq:main_ineq_fid_scaling}
  \sqrt{ \frac{4 (d-1)}{d (d+1)} } \times \sqrt{\upsilon}
  \geq \zeta^{\rm lub} (\upsilon) \vert_d
  \geq \sqrt{\frac{1}{d(d+1)}} \times \sqrt{\upsilon}.
\end{equation}
\end{proof}

\begin{ex}
\label{ex:high_fid_high_error_1}
All single-qubit unitary errors satisfy
\begin{equation}
\eta = \frac{1}{2}\eta^{\rm ub} = \sqrt{\frac{3}{2}(1-\varphi)}.
\end{equation}
If $\mathcal{D}_G (\rho) = U\rho U^\dagger$ for
some $2\times 2$ unitary operator $U$, then the eigenvalues of $U$ can be
written as $e^{\pm i \theta/2}$ for some $\theta$. The diamond distance
$d_\diamond$ and the fidelity are unitarily invariant, so the error rate of
$\mathcal{D}_G$ depends only on $\theta$. Furthermore, $\mathcal{D}_G$ is
equivalent to the generalized controlled-phase gate
(Definition~\ref{defn:c-phase}) and hence $\eta$ satisfies
Eq.~(\ref{eq:fidelity_scaling_error_rate}). Eq.~(\ref{eq:fidelity_scaling_fid})
therefore implies that $\eta^{\rm ub} = 2\eta$.
\end{ex}

\begin{ex}
\label{ex:high_fid_high_error_2}
There exists a two-qubit gate with fidelity $99.0\%$ but error rate $12.9\%$.
Consider the generalized controlled-phase gate (Definition~\ref{defn:c-phase})
acting on two qubits: one target qubit and one control qubit. Setting
$\theta=0.259$, we have $\varphi = 99.0\%$ by
Eq.~(\ref{eq:fidelity_scaling_fid}) and $\eta = 12.9\%$ by
Eq.~(\ref{eq:fidelity_scaling_error_rate}).
\end{ex}

We now demonstrate that the generalized controlled-phase gate example used to
prove Proposition~\ref{prop:fidelity_scaling} does not yield the true value of
$\eta^{\rm lub}$. We prove that $\zeta^{\rm lub}(d) \vert_\varphi \in
\Theta \left(d^{-1}\right)$, whereas the generalized controlled-phase gate has
$\zeta(d) \vert_\varphi \in \Theta\left(d^{-\frac{1}{2}}\right)$ by
Eq.~(\ref{eq:main_ineq_fid_scaling}).

\begin{prop}
\label{prop:dim_scaling}
For fixed fidelity $\varphi$,
\begin{equation}
  \zeta^{\rm lub}(d) \vert_\varphi \in \Theta\left(d^{-1}\right).
\end{equation}
Therefore,
\begin{equation}
  \zeta^{\rm ub}(d) \vert_\varphi \in \Theta\left(
    \zeta^{\rm lub}(d) \vert_\varphi
  \right).
\end{equation}
\end{prop}

\begin{proof}
We consider a special case of the generalized controlled-phase gate in which
$\theta=\pi$, so the unitary $U_\pi$ has an eigenvalue of $-1$. In this case,
$\|\mathcal{G}_{\rm id} - \mathds{1}\|_\diamond = 2$ by Theorem~26
of~\cite{JKP09}. The implementation we consider is
\begin{equation}
  \mathcal{G}_{\rm ac} := (1-\lambda) \mathcal{G}_{\rm id} + \lambda \mathds{1}.
\end{equation}
The error rate is
\begin{equation}
  \label{eq:dim_scaling_error}
  \eta
  = \frac{1}{2} \left\|
      \mathcal{G}_{\rm ac} - \mathcal{G}_{\rm id}
    \right\|_\diamond
  = \lambda \times \frac{1}{2} \left\|
      \mathcal{G}_{\rm id} - \mathds{1}
    \right\|_\diamond
  = \lambda.
\end{equation}
We calculate the fidelity by applying Nielsen's formula~\cite{Nie02} to the
Kraus decomposition
\begin{equation}
  \left\{ \sqrt{1-\lambda} I, \sqrt{\lambda} U_\pi \right\}
\end{equation}
of the discrepancy channel $\mathcal{D}_{U_\pi}$:
\begin{equation}
  \label{eq:dim_scaling_fid}
  \varphi = \frac{
    d + (1-\lambda) \vert\Tr(I)\vert^2 + \lambda \vert\Tr(U_\pi)\vert^2
  }{d+d^2} = 1 - \frac{4(d-1)}{d(d+1)} \times \lambda.
\end{equation}
Combining Eq.~(\ref{eq:dim_scaling_error}) with Eq.~(\ref{eq:dim_scaling_fid})
yields
\begin{equation}
  \zeta = \frac{4(d-1)}{d(d+1)}\upsilon.
\end{equation}
By definition, $\zeta^{\rm ub} \leq \zeta^{\rm lub} \leq \zeta$, which implies
\begin{equation}
  \frac{1}{d} \sqrt{\frac{\upsilon}{1+\frac{1}{d}}}
  \leq \zeta^{\rm lub} \leq \frac{4(d-1)}{d(d+1)}\upsilon.
\end{equation}
For fixed $\upsilon$, we define the constants $c_1 = 2^{-\frac{1}{2}}$ and $c_2
= 4\upsilon$. Then
\begin{equation}
  \frac{c_1}{d} \leq \zeta^{\rm lub}(d)\vert_\upsilon \leq \frac{c_2}{d},
\end{equation}
hence $\zeta^{\rm lub}(d) \in \Theta(d^{-1})$.
\end{proof}

We have established that $\eta^{\rm ub}$ is asymptotically tight with respect to
fidelity (Proposition~\ref{prop:fidelity_scaling}) and dimension
(Proposition~\ref{prop:dim_scaling}) if the other is fixed. Furthermore, we
showed that $\eta^{\rm ub}$ differs from the tightest possible bound by at most
a factor of $2\sqrt{d-1}$, where $d$ is the dimension of the gate. Although we
conjecture that $\eta^{\rm ub}$ is indeed the tightest possible bound on error
rate of a $d$-dimensional gate based only upon fidelity, important quantitative
statements are true (Examples~\ref{ex:high_fid_high_error_1}
and~\ref{ex:high_fid_high_error_2}) even if our conjecture is false.

%% PAULI DISTANCE %%%%%%%%%%%%%%%%%%%%%%%%%%%%%%%%%%%%%%%%%%%%%%%%%%%%%%%%%%%%%%

\section{New bounds for approximate Pauli channels}
\label{sec:Pauli_distance}

Here we derive improved bounds based on an additional promise about noise.
Specifically, we provide alternative lower and upper bounds on error rate in
terms of gate fidelity and a quantity we call the ``Pauli distance''. We show
that the connection between error rate and gate fidelity is strongly improved if
the Pauli distance of the noise process is known. We give numerical examples for
two important single-qubit noise processes: amplitude damping and unitary error.

The Pauli distance is defined to be the diamond distance between a channel and
its Pauli-twirl. To be precise, we define the single-qubit Pauli operators as
the unitary matrices
\begin{equation}
  I := \left(
    \begin{matrix} 1 & 0 \\ 0 & 1 \end{matrix}
  \right),\ 
  X := \left(
    \begin{matrix} 0 & 1 \\ 1 & 0 \end{matrix}
  \right),\ 
  Y :=  \left(
    \begin{matrix} 0 & -i \\ i & 0 \end{matrix}
  \right), \text{ and }
  Z := \left(
    \begin{matrix} 1 & 0 \\ 0 & -1 \end{matrix}
  \right);
\end{equation}
a multi-qubit Pauli operator is a tensor-product of single-qubit Pauli
operators. A Pauli channel is a quantum channel that has a Kraus representation
in which each Kraus operator is proportional to a Pauli operator.

\begin{defn}
The \emph{Pauli-twirl} of an $n$-qubit channel $\mathcal{E}$ (i.e.~$d=2^n$) is
\begin{equation}
  \mathcal{E}^{\rm PT} (\bullet)
  := \frac{1}{4^n} \sum_{k=1}^{4^n}
    P_k^\dagger\,\mathcal{E}\left(P_k \bullet P_k^\dagger\right) P_k,
\end{equation}
where $P_k$ represents a choice of $n$-qubit Pauli operator.
\end{defn}

\begin{defn}
We define the \emph{Pauli distance} of a gate implementation with discrepancy
channel
$\mathcal{D}_G$ to be
\begin{equation}
  \delta^{\rm Pauli} := d_\diamond \left(\mathcal{D}_G, \mathcal{D}_G^{\rm PT}\right),
\end{equation}
where $\mathcal{D}_G^{\rm PT}$ is the Pauli-twirl of $\mathcal{D}_G$.
\end{defn}

The Pauli-twirl of any channel is a Pauli channel, and the Pauli-twirl of a
Pauli channel is the same channel. For any channel $\mathcal{E}$, $\mathcal{E}$
and $\mathcal{E}^{\rm PT}$ have the same average gate fidelity as the average
gate fidelity is linear and invariant under unitary conjugation. The diamond
distance for channels of a fixed fidelity is minimized by Pauli channels, which
satisfy $\eta^{\rm Pauli} = (1+2^{-n}) (1-\varphi)$ where $n$ is the number of
qubits~\cite{MGE12,WF14}. Several common sources of noise, such as depolarizing
error and dephasing ($T_2$) processes, can be represented by Pauli
channels~\cite{NC00}. Such noise processes have $\delta^{\rm Pauli} = 0$. Other
sources of noise, such as amplitude-damping processes and unitary
errors, cannot. In these cases, $\delta^{\rm Pauli}>0$.

\begin{prop}
\label{prop:Pauli_distance}
The error rate $\eta$ of an $n$-qubit gate with gate fidelity $\varphi$ and
Pauli distance $\delta^{\rm Pauli}$ satisfies
\begin{equation}
\left\vert \delta^{\rm Pauli} - \eta^{\rm Pauli} \right\vert \leq \eta \leq
\delta^{\rm Pauli} +\eta^{\rm Pauli}.
\end{equation}
\end{prop}

\begin{proof}
By the triangle inequality,
\begin{equation} 
\label{eq:triangle_ineq}
 \frac{1}{2}\left\|\mathcal{D}_G
  - \mathds{1}\right\|_\diamond =
  \frac{1}{2}\left\|\mathcal{D}_G - 
  \mathcal{D}_G^{\rm PT} +
  \mathcal{D}_G^{\rm PT} -
  \mathds{1}\right\|_\diamond
  \leq
  \frac{1}{2}\left\|\mathcal{D}_G - 
  \mathcal{D}_G^{\rm PT}\right\|_\diamond +
  \frac{1}{2}\left\|\mathcal{D}_G^{\rm PT} -
  \mathds{1}\right\|_\diamond.
\end{equation}
The left-hand side equals $\eta$ and the right-hand side equals
$\delta^{\rm Pauli} + \eta^{\rm Pauli}$. Similarly, the reverse triangle
inequality implies $\left\vert \delta^{\rm Pauli} - \eta^{\rm Pauli} \right\vert
\leq \eta$.
\end{proof}

Proposition~\ref{prop:Pauli_distance} thus enables bounds to be placed on
possible values of $\eta$ in terms of $\varphi$ and $\delta^{\rm Pauli}$.
Indeed, a variation of this proposition can be applied to noise channels that
have a known structure.

\begin{prop}
\label{prop:sum_of_Pauli_distances}
Suppose $\mathcal{D}_G = \sum_k \mathcal{E}_k$, where each $\mathcal{E}_k$ is
some quantum channel. Let $\delta^{\rm Pauli}_k$ represent the Pauli distance of
$\mathcal{E}_k$. Then the error rate $\eta$ of $\mathcal{D}_G$ satisfies
\begin{equation}
  \eta \leq \eta^{\rm Pauli} + \sum_k \delta^{\rm Pauli}_k.
\end{equation}
\end{prop}

\begin{proof}
If $\delta^{\rm Pauli}$ is the Pauli distance of $\mathcal{D}_G$,
Proposition~\ref{prop:Pauli_distance} implies that
\begin{equation}
  \eta \leq \eta^{\rm Pauli} + \delta^{\rm Pauli},
\end{equation}
so we only need to show that
\begin{equation}
  \delta^{\rm Pauli} \leq \sum_k\delta^{\rm Pauli}_k.
\end{equation}
As the Pauli-twirl operation on quantum channels  is linear, i.e.
\begin{equation}
  \mathcal{D}_G^{\rm PT} = \left(\sum_k \mathcal{E}_k\right)^{\rm PT}
  = \sum_k \mathcal{E}_k^{\rm PT},
\end{equation}
we apply the triangle inequality repeatedly to obtain
\begin{equation}
  \frac{1}{2}\left\|
    \mathcal{D}_G - \mathcal{D}_G^{\rm PT}
  \right\|_\diamond
  = \frac{1}{2} \left\|
    \left( \sum_k \mathcal{E}_k \right)
    - \left( \sum_k \mathcal{E}_k^{\rm PT} \right)
  \right\|_\diamond
  \leq \sum_k \frac{1}{2} \left\|
    \mathcal{E}_k - \mathcal{E}_k^{\rm PT}
  \right\|_\diamond.
\end{equation}
The left-hand side equals $\delta^{\rm Pauli}$ and the right-hand side equals
$\sum_k \delta^{\rm Pauli}_k$.
\end{proof}

Although Proposition~\ref{prop:sum_of_Pauli_distances} yields weaker bounds than
Proposition~\ref{prop:Pauli_distance} in general, it might be easier in practice
to estimate $\delta^{\rm Pauli}$ for individual sources of noise rather than for
the overall noise process.

\begin{figure}
  \includegraphics{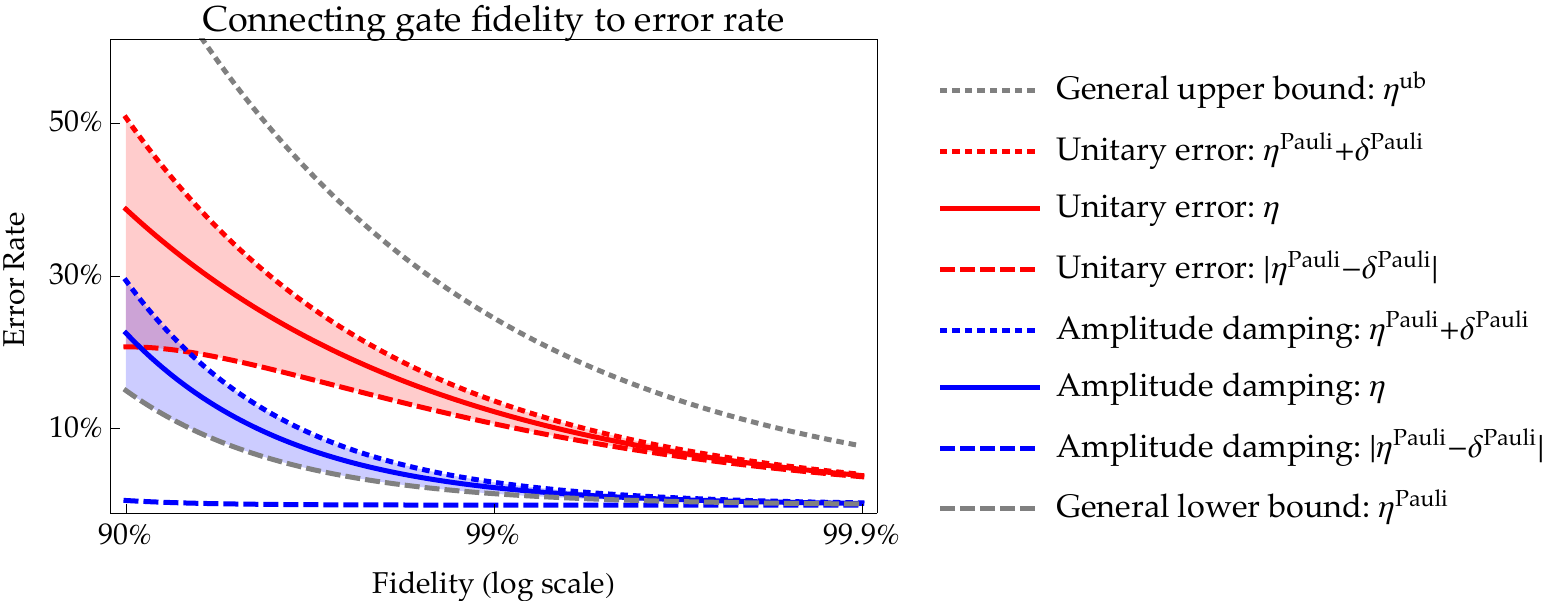}
  \caption{
    An illustration of the dichotomy between average gate fidelity and the error
    rate for single-qubit noise channels. The grey curves illustrate the
    generally applicable lower (dashed) and upper (dotted) bounds.
    The red curves pertain to unitary errors and
    the blue curves pertain to an amplitude damping (``a.d.'') process.
    The solid red/blue curves are the numerical values of the error rate
    (vertical axis) given the average gate fidelity (horizontal axis) of the
    unitary/a.d. model.
    The dotted red/blue curves are the values of the
    Pauli-distance-based upper bound $\eta^{\rm Pauli} + \delta^{\rm Pauli}$
    calculated for the unitary/a.d. model;
    the dashed are values for the lower
    bound $\left\vert\eta^{\rm Pauli} - \delta^{\rm Pauli}\right\vert$. The
    red/blue shading indicates region estimates for error rate
    based upon fidelity and Pauli distance; as $\left\vert\eta^{\rm Pauli} -
    \delta^{\rm Pauli}\right\vert$ falls below $\eta^{\rm Pauli}$ for the
    a.d. process, the blue region uses $\eta^{\rm Pauli}$ rather than
    $\left\vert\eta^{\rm Pauli} - \delta^{\rm Pauli}\right\vert$
    as a lower bound for $\eta$.
    The calculations were performed using the QETLAB project~\cite{QETLAB}.
  }
  \label{fig:Pauli_distance}
\end{figure}

We consider two examples of single-qubit noise processes in which
$\delta^{\rm Pauli}$ is non-zero: unitary error, which is a model of control
error, and an amplitude damping process,
which is a model of thermalization with a
zero-temperature bath. The unitary error can be entirely specified by the
eigenvalues $e^{\pm i \theta}$ of the unitary operator, and the
amplitude damping process
may be specified by a rate parameter $r$. Both $r$ and $\theta$ may be expressed
in terms of the observed average gate fidelity $\varphi$ and thus the error rate
of each can be numerically evaluated as a function of $\varphi$. The results of
this numerical evaluation are displayed in Figure~\ref{fig:Pauli_distance}.

The most important observation about Figure~\ref{fig:Pauli_distance} is that the
Pauli-distance-based bounds on $\eta$ yield excellent estimates of the error
rate of a noise process as fidelity increases. In fact, fidelity indicates
confidence interval if $\delta^{\rm Pauli}$ is considered as an estimate of
$\eta$. Therefore, the Pauli distance can be interpreted as a measure of the
`badness' of noise in the sense that it indicates the size of the gap between
fidelity and error rate.

%% DISCUSSION %%%%%%%%%%%%%%%%%%%%%%%%%%%%%%%%%%%%%%%%%%%%%%%%%%%%%%%%%%%%%%%%%%

\section{Assessing progress towards fault-tolerant quantum computing}
\label{sec:discussion}

The Threshold Theorem guarantees the possibility of fault-tolerant quantum
computation in the presence of local errors that occur at a rate $\eta$ below a
threshold value $\eta_0$. Our aim has therefore been to convert gate fidelity
$\varphi$, a commonly reported figure-of-merit for quantum logic operations,
into an upper bound $\eta^{\rm ub}$ that can be compared, in principle, to
$\eta_0$. Of course the noise assumptions underlying the Threshold Theorem could
be either weaker or stronger than reasonable assumptions about the noise of real
devices, but this subtlety is often overlooked: numerical simulations such as
those of Knill~\cite{Kni05} and Raussendorf and Harrington~\cite{RH07} are often
considered to be indicative of a code-specific threshold value $\eta_0^{\rm lb}$
even though both papers are clear that only one well-behaved noise model is
being simulated.

The proper interpretation of these results is, in the words of Knill, as
``evidence that accurate quantum computing is possible for [error rates] as high
as three per cent''. Thus, Knill claims not that $3\%$ is an \emph{estimate} of
$\eta_0^{\rm lb}$ for the C4/C6 code, but that it is an \emph{upper bound}. The
results of Raussendorf and Harrington can be interpreted similarly. As we stated
at the end of Sec.~\ref{subsec:threshold_thm}, the connection between such
simulations and the estimation of threshold values for actual devices is the
subject of ongoing research~\cite{STD05,MPGC13,PGH+14}.

Whatever its actual value, the threshold error rate that is guaranteed to exist
by the Threshold Theorem is often treated as a performance target for research
efforts towards fault-tolerant quantum computing~\cite{BWC+11,CGC+12,BKM+14}.
However, these authors quote the threshold not as a target error rate but as a
target average gate fidelity. As we have shown in this paper,
the error rate of a quantum gate cannot, in general,
be computed as a function of fidelity.
Therefore, the kind of threshold demonstrated to exist
by the Threshold Theorem is not a fidelity threshold.

Of course we have agreed that bounds on the error rate of a quantum gate can be
derived from the average gate fidelity~\cite{WF14,MGE12}. Whereas the lower
bound was already known to be tight, we showed in Sec.~\ref{sec:tightness} that
the upper bound is an asymptotically tight approximation to the tightest
possible upper bound. We also agreed that
the quantum gate error rate can be computed as a function of fidelity
if the noise is guaranteed to be Pauli;
indeed, we showed in Sec.~\ref{sec:Pauli_distance} that this relationship is
approximately true if the noise can be represented by an approximate-Pauli
channel. But noise is demonstrably non-Pauli in
experiments~\cite{HCD+12,CVC+13,CNR+14,OKB+15} so the observed average gate
fidelity is not necessarily indicative of the true error rate. Existing
threshold results do not imply a practical performance target in terms of gate
fidelity.

The wide-spread conflation of average gate fidelity with error rate
has led to assertions that threshold fidelities for Pauli noise
correspond to fidelity targets for general noise.
One group~\cite{CGC+12}, for example, claims
that gate fidelities of
90-99.5\% (``depending on measurement errors'') suffice for fault-tolerant
quantum computation using the surface code. This is only known to be true if the
relevant noise model is Pauli, which it is not.
Another group~\cite{BKM+14} goes further by asserting
that device performance has surpassed the
fault-tolerance threshold for surface-code-based quantum computing. Their stated
threshold value is $99\%$ fidelity, which is derived from simulations of the
code in the presence of depolarizing noise~\cite{Fow13}. Yet depolarizing noise
is Pauli and therefore saturates the lower bound on error rate as a function of
fidelity, and the appendix of~\cite{BKM+14} makes it clear that there are
non-Pauli sources of noise. Even if the quoted threshold value of $1\%$ is
trustworthy, it is a threshold error rate and not a threshold infidelity.

If the threshold error rate is indeed $1\%$, then $\eta^{\rm ub}$ yields
rigorous, but relatively pessimistic, fidelity targets. If $\eta_0^{\rm est}$ is
the error rate to be surpassed, a gate fidelity satisfying
\begin{equation}
  \label{eq:thresh_condition}
  \varphi > 1- \frac{\left(\eta_0^{\rm est}\right)^2}{d(d+1)},
\end{equation}
where $d$ is the dimension of the gate, is required to guarantee an error rate
$\eta$ below $\eta_0^{\rm est}$ without additional information. So if we assume
that $\eta_0^{\rm est}=1\%$, two-qubit gates ($d=4$) need to have a fidelity
greater than $1-5\times 10^{-6} = 99.9995\%$ to ensure that the error rate falls
below $1\%$. It is of course possible that lower fidelities suffice, but such a
claim must be defended with information such as the Pauli distance
(Sec.~\ref{sec:Pauli_distance}) or unitarity~\cite{WGHF15} about gate
performance additional to fidelity; the main point of our paper is that such
additional information is \emph{required}.

%% CONCLUSION %%%%%%%%%%%%%%%%%%%%%%%%%%%%%%%%%%%%%%%%%%%%%%%%%%%%%%%%%%%%%%%%%%

\section{Conclusion}
\label{sec:conclusion}

Reports of extremely high average gate fidelities engender optimism that current
technology is near the threshold required for fault-tolerant quantum
computation. Yet, although the average gate fidelity $\varphi$ is an
experimentally convenient figure of merit, it is not the proper metric, i.e.~the
worst-case quantum gate error rate $\eta$, for assessing progress towards
fault-tolerance.

We have shown that $\eta^{ub} = \sqrt{d(d+1)(1-\varphi)}$ is an asymptotically
tight estimate of the tightest possible upper bound to $\eta$ in terms of
$\varphi$ alone, and we conjecture that this is optimal. We have demonstrated
that it is possible for a two-qubit gate with $99\%$ fidelity to have an error
rate of nearly $13\%$, and we have demonstrated that fidelity-based assessments
of gate performance underestimate especially important noise sources such as
unitary error. We have derived an alternative bound that can yield tighter
estimates of gate performance if an additional piece of information we call the
``Pauli distance'' of the noise channel is known, though other kinds of
information can also be used to derive alternative bounds
\cite{WGHF15,KLDF15,Wal15}.

We have given a sobering assessment of reported progress towards fault-tolerant
quantum computation by converting reported average gate fidelity to the
worst-case quantum gate-error rate. Based on the best theoretical results
currently available, we have shown that two-qubit gates must surpass $99.9995\%$
gate fidelity to ensure that gates experience an error rate lower than $1\%$.
We have used the Pauli distance to show that information additional to fidelity
can be employed to justify tighter bounds on gate performance,
and we argue that future attempts to verify quantum gate performance
should include estimates of figures of merit additional to fidelity
in order to circumvent the looseness of fidelity-based bounds.

\subsection*{Acknowledgements}
We appreciate valuable discussions with
R.~Blume-Kohout,
K.~Brown,
J.~Emerson,
A.~Fowler,
D.~Gottesman,
C.~Granade,
P.~Groszkowski,
R.~Laflamme,
M.~Mosca, and
J.~Watrous.
YRS acknowledges financial support from
the Office of the Director of National Intelligence (ODNI),
Intelligence Advanced Research Projects Activity (IARPA),
through the U.S. Army Research Office.
JJW acknowledges financial support from
the U.S. Army Research Office through grant W911NF-14-1-0103.
BCS acknowledges financial support from
the Natural Sciences and Engineering Research Council of Canada,
Alberta Innovates Technology Futures, and
China's 1000 Talent Plan.
All statements of fact, opinion, or conclusions contained herein
are those of the authors and should not be construed
as representing the official views or policies of
IARPA, the ODNI, or the U.S. Government.

%% EOF %%%%%%%%%%%%%%%%%%%%%%%%%%%%%%%%%%%%%%%%%%%%%%%%%%%%%%%%%%%%%%%%%%%%%%%%%

\section*{References}
\bibliography{gateerror}

\end{document}